\documentclass[pdflatex,sn-mathphys-num]{sn-jnl}

\usepackage{graphicx}%
\usepackage{multirow}%
\usepackage{amsmath,amssymb,amsfonts}%
\usepackage{amsthm}%
\usepackage{mathrsfs}%
\usepackage[title]{appendix}%
\usepackage{xcolor}%
\usepackage{textcomp}%
\usepackage{manyfoot}%
\usepackage{booktabs}%
\usepackage{algorithm}%
\usepackage{algorithmicx}%
\usepackage{algpseudocode}%
\usepackage{listings}%

\usepackage{parskip}
\usepackage{amssymb}
\usepackage{xtab}
\usepackage{float}
\usepackage{enumitem}
\usepackage{mathalfa}
\usepackage{array}
\usepackage{dsfont}

\usepackage{wrapfig}
\usepackage{cases}
\usepackage{bm}
\usepackage{tabularray}
\usepackage{hyperref}
\usepackage{cleveref}
\usepackage{subcaption}

\theoremstyle{thmstyleone}%
\newtheorem{theorem}{Theorem}
\newtheorem{proposition}[theorem]{Proposition}%
\newtheorem{lemma}[theorem]{Lemma}%

\theoremstyle{thmstyletwo}%

\theoremstyle{thmstylethree}%
\newtheorem{definition}{Definition}%

\begin{document}

\title[Article Title]{Convergence of Replicator Dynamics in the Repeated Prisoner's Dilemma with Restarts}


\author*[1]{\fnm{Benedict} \sur{Russell}}\email{benedict.i.russell@warwick.ac.uk}

\author[1]{\fnm{Chin-wing} \sur{Leung}}\email{chin-wing.leung@warwick.ac.uk}

\author[1]{\fnm{Paolo} \sur{Turrini}}\email{p.turrini@warwick.ac.uk}

\affil[1]{\orgname{University of Warwick}, \orgaddress{\city{Coventry}, \country{UK}}}

\abstract{
We investigate a population of self-interested agents playing a repeated Prisoner's Dilemma under the trigger-restart mechanism. Under such a mechanism, agents play a sequence of symmetric games with their partner, and restart the interaction if their actions disagree. Our work focuses on the convergence of replicator dynamics in a well-mixed population of agents, where the emergence of cooperation is challenged by the individual incentive for exploitation. Formulating the corresponding parametrised normal‐form game, with agents each adopting a length-$m$ strategy sequence, we show that increasing the strategy length enables cooperation to emerge and stabilise. We provide exact convergence guarantees for restricted strategy lengths and, in the general payoff configuration, provide the necessary parametric conditions for the stability of cooperative strategies. By deriving an exact formula for the number of stable sequences, we find structural properties necessary for stability, as agents must learn to initially defect - the so-called ``hazing period" - before cooperating indefinitely. Our analysis shows that, while optimal cooperative sequences exist, agents favour less-optimal sequences with a longer hazing period, which possess larger basins of attraction. 
}


\maketitle

\section{Introduction}
From resource distribution and tax compliance to collective cell migration, financial markets and peer-to-peer networks, multi-agent systems can often only function if individual benefits are put aside in favour of a greater social good. This is stylised in games like the Prisoner's Dilemma (PD), where the benefits of mutual cooperation are overshadowed by the individual incentive to exploit others. If left to their own individual decisions, self-interested agents are bound to settle on inefficient outcomes, even ending up in what is known as the ``Tragedy of the Commons" \cite{Hardin_68}.

Social interaction can act as a mechanism to promote cooperation, where agents develop long-term relationships to build trust and achieve mutual benefit. One could use harsh punishment mechanisms to outcast defectors, such as the grim-trigger, where a player cooperates as long as their opponent does but permanently defects after a single defection \cite{friedman1971non}. 
Memory-based strategies have been extensively studied, popularised by the Tit-for-Tat strategy (TFT), which begins with cooperation, and employs direct reciprocity (and retribution) by mirroring the opponent's last action \cite{axelrod}.

However, in multi-agent environments, mutual cooperation can fail as defectors can simply switch to another partner when the opportunity arises. 
To counter this, we consider the class of infinitely repeated two-player games, where interactions can restart \cite{berker_computing_2024, berker_asymmetric}. Players agree to a sequence of actions e.g. $(D,C,C,...)$ and if neither deviates, both receive the reward $(1,3,3,...)$, where payoffs are taken from Table \ref{table: payoff}. Otherwise, they obtain the corresponding reward for deviation and the sequence restarts from the beginning. For example, if one player switches to defection in the second round, their reward profiles become $(1, 5, 1, 5,...)$ and $(1, 0, 1, 0,...)$, with the sequence restarting at round 2 each time. To reach mutual cooperation, players must coordinate over a period of lower payoffs before adopting a cooperative strategy.

While we know the trigger-restart mechanism is consistent with cooperative equilibria, we do not know if such equilibria can be reached by distributed local interactions in large agent populations, as well as how sensitive they are to the initial conditions. In this paper, we bridge this gap, using tools from dynamical systems to determine when and where the population will converge, providing explicit conditions on the game structure for cooperative equilibria to emerge through the trigger-restart mechanism.

\textbf{Contribution}.
We study the evolutionary dynamics of a well-mixed population of agents playing repeated Prisoner’s Dilemma with restarts.
We formulate the game as a parametrised normal‐form game, where agents each adopt a length-$m$ strategy sequence, with their final action indefinitely repeated.
We show that when the temptation payoff is large, pure defection is the unique Nash equilibrium for short strategy lengths.
We prove the necessary conditions for cooperative equilibria to exist, and show that pure defection leaves the population in the limit of $m$. Delving deeper into the structure of the stable sequences, we prove exponential stability of strict Nash equilibria, and derive an exact formula for the number of such sequences.
We analyse the basins of attraction, and show that while optimal cooperative sequences exist, agents favour longer hazing periods. 

\textbf{Paper Structure.}
Section \ref{sect: literature} presents the related literature. Section \ref{sect: preliminaries} discusses the preliminaries. Section \ref{sect: formulation} outlines the game formulation and necessary background. Section \ref{sect: theory1} analyses the theoretical properties of sequences, while Section \ref{sect: theory2} looks at stability and basin distributions. We conclude with future directions in Section \ref{sect: discussion}.

\section{Related Literature}
\label{sect: literature}
The classical route to cooperation in repeated Prisoner's Dilemmas relies on direct reciprocity, which conditions current actions based on historical information, to sustain a cooperative population. Indeed, TFT's famed success in heterogeneous environments highlights this behaviour \cite{axelrod}. 
Beyond memory-one strategies, analysis of memory-two has uncovered the formation of symmetric reinforcement learning equilibria in a two-player scenario \cite{Ueda_2023}. Other threshold-based rules, such as all-or-none, condition behaviour on unanimous outcomes over the entire history \cite{Hilbe_2017}, where extensive analysis under varying memory length and tolerance to errors showed improved robustness, leading to more cooperative outcomes. These strands show that increasing the memory used to inform actions can offer a rich class of equilibria beyond that of TFT \cite{zhang2024}.
However, under anonymous random matching, the strong informational requirements for memory-based strategies ---specifically, access to their new partner's previous actions--- can still be insufficient for the emergence of cooperation \cite{anastassacos_partner_2020}. 
This has led to a search for mechanisms to provide additional structure to the environment.

Cooperation in repeated games, with the possibility of restarting interactions, has been studied in both multi-agent learning \cite{anastassacos_partner_2020, defection_russell2026, russell2026dynamicspolicygradient} and economic theory settings \cite{Fujiwara, IZQUIERDO201491}. Partner selection has been identified as one of the effective mechanisms in promoting cooperation under various studies, both empirically \cite{BarclayWiller2007partnerchoice, rand2011, wang2012, Zhang2016} and theoretically \cite{SegbroeckSNPL09, sylwester2010cooperators, zheng_simple_2017, BaraTA22, russell2026meanfieldimitationdynamicsfast, graser2025repeated}.
Recent contributions showed that agents can, in fact, co-learn this partner selection rule and retaliate against defection by looking for a new partner \cite{leung_learning_2024, LeungLT24}.
Reputation has been recognised as another effective mechanism \cite{sabater2002, pujol2002, perreaudepinninck2010, smit2024, ren2025bottom} which can lead to cooperation, but it relies on a shared understanding of what constitutes good or bad actions, as well as effective communication channels \cite{SantosPS18}.
The results highlight the importance of developing trust and slowly building a stable partnership.

The concept of sequence-based strategies, which restart upon deviation, embeds the idea of interaction-length dependence.
Recent contributions have formalised the computation of finding equilibria under the trigger-restart mechanism \cite{berker_computing_2024, berker_asymmetric}, showing that the welfare-optimal, and indeed any cooperative sequence, must necessarily begin with a `hazing period' of mutual defection followed by eventual cooperation. As a consequence, the temptation to deviate from the sequence is balanced with the threat of repeating the initial hazing. Our work differs from their algorithmic equilibrium analysis by asking if and when cooperation will endogenously arise under evolutionary pressures, and how the payoff parameterisation affects the properties of the defection sequences through the lens of population dynamics.

Replicator dynamics \cite{taylor1978evolutionary} provide a simple framework to model the evolutionary dynamics in large, well-mixed populations. It has been used to study the stability of cooperative strategies in repeated games, ranging from the symmetric ultimatum game \cite{miekisz_replicator_2012} to graph-based social dilemmas \cite{ohtsuki_replicator_2006}. Recent work on global convergence in the repeated snow-drift game partitioned the strategy simplex by construction of strategy frequency ratios, proving convergence for every interior and boundary trajectory \cite{ramazi_global_2021}. In this paper, we adopt similar geometric techniques but in a richer, sequence-based policy space. This leads to different equilibrium structures and complex basin geometries tied to the length of the hazing period.

\section{Preliminaries}
\label{sect: preliminaries}
We review repeated Prisoner's Dilemmas equipped with restart mechanisms. We then provide the necessary background on replicator dynamics and dynamical systems.
\subsection{Repeated Prisoner’s Dilemma with Restarts}
We consider an infinitely large, well-mixed population of agents playing a repeated PD. Agents are given two possible actions: cooperate or defect, denoted C and D respectively. The payoff matrix is given by the pairwise interaction and is symmetric for each player, shown in Table \ref{table: payoff}.

\begin{table}[h]
   \caption{Payoff matrix}
  \label{table: payoff}
  \centering
  \begin{tabular}{lll}
    \toprule
         & C & D \\
    \midrule
    C & ${R},{R}$  & ${S},{T}$     \\
    D     & ${T},{S}$ & ${P},{P}$      \\
    \bottomrule
  \end{tabular}
\end{table}

The payoffs are such that $T > R > P > S$, with the first entry in each cell received by the row player.
This constraint ensures that it formulates a Prisoner's Dilemma, where there is an incentive to defect. For simplicity, in this paper we let $S=0$ and allow $P,R,T$ to parametrise the payoff matrix. Following \cite{ramazi_global_2021}, we denote this as our base game, and denote it $B$. In the one-shot scenario, defection is the only Nash equilibrium, and so cooperation cannot be achieved. 
In the repeated PD with restarts, both players agree on a common sequence of actions, such as $(D, D, C, ...)$, to be played. Under the trigger-restart mechanism, agents will accumulate rewards along this sequence whilst they agree upon the actions. The sequence is restarted when the two actions differ.
In the two-player game, it has been shown that the socially optimal sequence consists of a hazing period where both agents defect, followed by infinitely repeating cooperation \cite{berker_computing_2024}. Thus, we consider that agents adopt a sequence of $m$ actions, where the last action repeats after the $m^{th}$ round. As $m\to\infty$, it covers all possible action sequences.

\subsection{Replicator Dynamics}
We study the evolutionary dynamics of the population under the repeated Prisoner's Dilemma (PD) with restarts. The fraction of the population using a given strategy $s_k$ is $0 \leq x_k \leq 1, \forall k\in \{1, ..., d\}$, where the dimension of the strategy space is given by $d= 2^m$. Define the simplex which contains all vectors of fractions $\mathbf{x} \in \Delta$ as
\begin{equation}\label{eq: simplex}
    \Delta = \{x \in \mathbb{R}^d : \sum_{k=1}^dx_k = 1, x_k \geq 0, k = 1,...,d\}
\end{equation}
The replicator equation is given by 
\begin{equation}\label{eq: replicator}
    \dot{x}_k = x_k[(A\mathbf{x})_k - \mathbf{x}^TA\mathbf{x}],\quad k \in \{1,...,d\}
\end{equation}
where $\dot{x}_k$ denotes the time derivative. This describes an evolutionary process whereby strategy $k$ will increase if it has above-average payoff and decrease if it has below-average payoff. In this context, the process can be understood as a mean-field approximation of the Moran process \cite{moran} with imitation \cite{nowak2004emergence}:
\begin{enumerate}
    \item Each agent adopts a fixed sequence of actions, of length $m$.
    \item Agents are randomly paired, and play their sequence against their opponent under the trigger-restart mechanism. This repeats with the whole population, and thereby the agent accumulates an average payoff.
    \item An agent $i$ picked uniformly at random picks another agent $j$ in the population; if their payoff is higher, $i$ switches strategy to $j$'s with a rate proportional to the payoff difference.
\end{enumerate}

\subsection{Dynamical Systems}
To provide a rigorous analysis of the evolutionary dynamics, we borrow tools from dynamical systems. We outline the necessary definitions required which can be found in \cite{khalil_nonlinear_2002}. The population evolution is modelled as a trajectory starting and ending in \eqref{eq: simplex}.

\begin{definition}\label{def: positively invariant}
     Let $\dot{x} = f(x)$. Then $U \subseteq \mathbb{R}^d$ is positively invariant if for every $x(0) \in U$ then $x(t) \in U$ for all $t \geq 0.$ 
\end{definition}
It is a standard result that $\Delta$ is positively invariant under \eqref{eq: replicator} \cite{Sigmund1986}. Intuitively, the strategy distribution across the population can never exceed the boundaries of being a probability distribution. 
Stability is one of the most important characteristics of a dynamical system, which refers to the robustness of the stationary points of a system under perturbation.

Exponential stability indicates that any small perturbation will exponentially decay back to the equilibrium point.
\begin{definition}\label{def: exponential stability}
     Let $x^*=0$ be a rest point of the system $\dot{x} = f(x)$, with $x(t_0)=x_0$, then $x^*$ is exponentially stable if there exist $\delta, k, \lambda >0$ such that 
     \[
        \|x(t)\|\leq k\|x_0\|e^{-\lambda(t-t_0)}
     \]
     for all $\|x_0\|<\delta$.
\end{definition}
It is known that if the Jacobian of the dynamics at the stationary point $x^*$ is Hurwitz, the system is exponentially stable at $x^*$.
\begin{lemma}\label{lemma: exponential stability}\cite{khalil_nonlinear_2002}
     Let $x^*=0$ be a rest point of the system $\dot{x} = f(x)$. Define $J=[\partial f/\partial x](0)$. Then $x^*$ is an exponentially stable point iff the eigenvalues of $J$ satisfy Re$(\lambda_i) <0$ (Hurwitz). 
\end{lemma}
In the bounded domain of the simplex, this stability can be assessed on the tangent space $T_x\Delta =\{v \in \mathbb{R}^d: \sum_i v_i = 0\}$. 
Complementing this notion of stability is the proportion of the strategy space which converges to these fixed points; this region is known as the basin of attraction (BoA).
\begin{definition}\label{def: basin of attraction}
     Let $\phi(t;x)$ be the solution to $\dot{x}=f(x)$, and $x^*$ be an equilibrium point. Then the basin of attraction for $x^*$ is the set of all points $x$ such that the solution exists for $t\geq0$ and $\lim_{t\rightarrow\infty} \phi(t;x) = x^*$. 
\end{definition}
In high dimensions, analytical formulations for each basin are often intractable. Researchers can consider asymptotic analysis under certain regularity conditions such as ordinal potential games \cite{collevecchio2024basins}. When such guarantees are not in place, numerical schemes can be used to estimate and analyse the basin sizes.

\section{Repeated Games with Restarts: Learning Finite Strategies}\label{sect: formulation}

In this section, we formulate repeated games with restarts and study a simplified setting where agents are learning finite strategy sequences made by $m$ actions, where the last action repeats after the $m^{th}$ round. Besides ensuring tractability, this assumption allows for a direct comparison with the results in the literature \cite{berker_computing_2024, berker_asymmetric}.
Concretely, define
\begin{align*}
    \mathcal{S} := \{s :s_i \in \{C,D\},\;\forall  i \in \{1,...,m\}\},
\end{align*}
as the set of all possible sequences where the action $s_m$ repeats after round $m$. Increasing the length of $m$ will allow for greater strategic complexity, creating longer sequences for the agents to decide between.
At the start, each agent adopts a length $m$ sequence in $\mathcal{S}$. Agents are then paired; let $s,t\in\mathcal{S}$ be the sequences of a given pair. Define the time of the first deviation between the two sequences as \[
\tau(s,t) \;=\;
\begin{cases}
\min\{\,i\in\{1,\dots,m\}:s_i\neq t_i\}, & s\neq t,\\
\infty, & s=t.
\end{cases}
\]
 The rewards are discounted at some rate $\gamma \in (0,1)$, enabling the construction of the normal form game's payoff matrix by ensuring a finite total reward for each possible sequence. For sequences $s,t \in S$, the payoff for the row player (adopting sequence $s$) is evaluated as
\begin{equation}\label{eq: payoff_nf}
A(s, t) :=
\begin{cases}
\dfrac{\displaystyle\sum_{i=1}^{\tau(s,t)} \gamma^{\,i-1}\,B_{s_i,t_i}}
{1 - \gamma^{\tau(s,t)}}, 
& \tau(s,t)<\infty, \\[1em]
{\displaystyle\sum_{i=1}^{m} \gamma^{\,i-1}\,B_{s_i,s_i}}
 {\;+\;
\dfrac{\gamma^m}{1 - \gamma}\,B_{s_m,s_m}},
& \tau(s,t)=\infty,
\end{cases}
\end{equation}
where $B_{s_i,t_i}$ is the immediate payoff for the row player when actions $(s_i,t_i)$ are selected and $\tau$ encodes if and when the trigger-restart mechanism is used. When $\tau < \infty$, the sequence of agreed-upon actions up to $\tau$ repeats indefinitely (actions will differ at step $\tau$ for each restart). Otherwise, the sequence continues with the final agreed-upon actions. Table \ref{tab:nfg_infinite} illustrates a concrete instantiation of the normal-form game with $m=2$.

\begin{table}
    \caption{Normal-form payoff ($m=2$)}
    \centering
    \renewcommand{\arraystretch}{1.5}
    \begin{tabular}{ccccc}
        \toprule
        & \text{CC} & \text{CD} & \text{DC} & \text{DD}\\ 
        \midrule
        \text{CC} & $\frac{3}{1-\gamma}$ & $\frac{3}{1-\gamma^2}$ & 0 &  0\\
        \text{CD} & $\frac{3 + \gamma 5}{1-\gamma^2}$ & $3+ \frac{\gamma}{1-\gamma}$ & 0 & 0\\
        \text{DC} & $\frac{5}{1-\gamma}$ & $\frac{5}{1-\gamma}$ & $1+\frac{3\gamma}{1-\gamma}$& $\frac{1}{1-\gamma^2}$\\
        \text{DD} & $\frac{5}{1-\gamma}$ & $\frac{5}{1-\gamma}$ & $\frac{1 + \gamma 5}{1-\gamma^2}$ & $\frac{1}{1-\gamma}$\\
        \bottomrule
    \end{tabular}
    \label{tab:nfg_infinite}
\end{table}

Humans were found to adopt strategies similar to TFT \cite{zheng_simple_2017}, suggesting that a restriction on $m$ may be more applicable to real-world scenarios in which agents adopt simpler strategies. In the context of social dilemmas, how one interacts with another is likely to be based only on the first few interactions, and once a pattern is established, the decision-making is set. In this paper, we will analyse both the minimal length requirements for cooperative behaviour to emerge and the asymptotic analysis.

Whilst condensing the game with geometric summations enables an explicit formulation of the normal form game, we are unable to utilise tools often used in analysing multi-agent games and their corresponding basins of attraction. As the next result shows, the formulation is not a potential game. Hence, no global potential function can be studied. 

\begin{lemma}\label{lem: potential game}
    For any initialisation of $A$ with $m\geq 2$, the game $\Gamma =\langle N, \mathcal{S}, A\rangle$ is not a potential game.
\end{lemma}
\begin{proof}
    Let $i,j \in N$ be the agents playing the game, where $A^i(\cdot,\cdot)$ and $A^j(\cdot,\cdot)$ denote the payoff utility for $i$ and $j$ respectively. $\Gamma$ is an exact potential game iff $\forall s,s',t,t' \in \mathcal{S}$ \cite{shapley1996},  
    \begin{align*}
        &[A^i(s',t) - A^i(s,t)] - [A^i(s',t') - A^i(s,t')] 
        = [A^j(s,t') - A^j(s,t)] - [A^j(s',t') - A^j(s',t)].
    \end{align*}
    Let $s=(C,C,...,C), s'=(D,C,...,C), t=(C,C,...,C), t'=(C,D,...,C)$. Then,
    \begin{align*}
        [\frac{T}{1-\gamma} - \frac{R}{1-\gamma}] - [\frac{T}{1-\gamma} - \frac{R+\gamma S}{1-\gamma^2}] &= [\frac{R+\gamma T}{1-\gamma^2} - \frac{R}{1-\gamma}] - [\frac{S}{1-\gamma} - \frac{S}{1-\gamma}]\\
        \Rightarrow\frac{\gamma(S-T)}{1-\gamma^2}&=0
    \end{align*}
    but $T>S$ and $\gamma \in (0,1)$, therefore $\Gamma$ is not a potential game.
\end{proof}
Instead, we utilise techniques at a local scale to analyse when specific sequences emerge which result in cooperative behaviour and find the sufficient conditions for the population to avoid converging to the defection Nash equilibrium.
\section{Convergence to Cooperation}
\label{sect: theory1}
In this section, we analyse the necessary and sufficient conditions for the population to learn cooperative stable sequences. We show pure defection dies out in the population as the strategy sequence length increases, and provide the explicit separation of the simplex when suitable restrictions are applied. We conclude by showing the parametrisation impact on the existence of cooperative sequences.

\subsection{Longer Sequences Avoid Pure Defection}
Since cooperation is not a Nash equilibrium in the one-shot game, we need sufficient degrees of freedom in the strategy space for it to emerge. Undoubtedly, the parametrisation of the payoff matrix will influence both the number of stable sequences and population dynamics. We begin by showing how simple, short sequences can be directly controlled by the payoff matrix structure.
\begin{lemma}\label{lem: m2_dd}
        Let $m=2$ and $\gamma\rightarrow1$. If $T+P\geq 2R$, then DD is the only Nash equilibrium.
\end{lemma}
\begin{proof}
    Note first that (DD) is a Nash equilibrium since there is no possible deviation which increases the payoff. By a similar proof to Lemma \ref{lem: num_nash}, a stable sequence must start with at least $\lfloor\frac{T-R}{R-P}+1\rfloor$ defections.
    \begin{align*}
        &T+P\geq 2R\\
        \Leftrightarrow\quad& \frac{T-R}{R-P} \geq1\\
        \Rightarrow{}\;\;& \lfloor\frac{T-R}{R-P}+1\rfloor \geq2.
    \end{align*}
    Hence, all stable sequences must start with two defections, which is exactly the length of the sequence.
\end{proof}
As we increase the strategy length $m$, more Nash equilibria are generated. It is not clear, however, how many steps are sufficient for the population to adopt a cooperative sequence. To provide a bound on when cooperative equilibria occur, we consider the closest strategy to pure defection. This strategy, which we denote the last-step cooperator (LC), exploits the stability of defection before utilising the infinite repetition of the final move to cooperate. Explicitly, these strategies are as follows:
\begin{align*}
    s_{LC} = (D,D,....,D,C), \qquad s_{D} = (D,D,....,D,D),
\end{align*}
with the first containing $m-1$ defections and the latter $m$.
Analysing the ratio of these strategies \cite{Hofbauer_Sigmund_1998}, we isolate the strategy dynamics of (LC) and pure defection (D),
\begin{equation}
    \frac{d}{dt}(\frac{x_{D}}{x_{LC}}) = \frac{x_{D}}{x_{LC}}[(A\mathbf{x})_{D} - (A\mathbf{x})_{LC}]
\label{eq: xd/xlc}
\end{equation}
and equate the fitnesses to find the payoff boundary. The separatrix can be derived by equating the payoffs:
\begin{align*}
    &(A\mathbf{x})_D - (A\mathbf{x})_{LC} = 0,\\
    \Leftrightarrow\quad &\sum_j A_{D,j}x_j - A_{LC,j}x_j =0,\\
    \Leftrightarrow\quad &x_{LC}(A_{LC,LC} - A_{D,LC}) + x_D(A_{LC,D} - A_{D,D}) =0,\\
    \Leftrightarrow\quad &x_{D} =x_{LC}\frac{A_{LC,LC} - A_{D,LC}}{A_{D,D} - A_{LC,D}}.
\end{align*}
The requirement for equality then reduces to analysing the quotient of the payoff differences. This gives the explicit formulation for the separatrix as
\begin{equation}\label{eqn: psi the seperatrix}
        \varphi(\gamma,B,m) = \frac{(R-P)\sum^{m-1}_{j=1}\gamma^j + (R-T)}{P}.
\end{equation}
This equation can be understood as the marginal utility of cooperation over the next $(m-1)$ steps, less the temptation to deviate. If there are insufficient steps, the temptation reward is worth repeating the initial phase of lower reward. Consequently, the dynamics of \eqref{eq: xd/xlc} can be concretely understood with the next lemma.
\begin{lemma}\label{lem: separatrix splits Ax}
    $(A\mathbf{x})_{LC} - (A\mathbf{x})_D <0$ if and only if $x_D > x_{LC}\; \varphi(\gamma,B,m)$.
\end{lemma}
\begin{proof}
Consider the payoff difference given by
\begin{align*}
     (A\mathbf{x})_{LC} - (A\mathbf{x})_D &< 0\\
    \Leftrightarrow\quad \sum_j A_{LC,j}x_j - A_{D,j}x_j &<0\\
    \Leftrightarrow\quad x_{LC}(A_{LC,LC} - A_{D,LC}) + x_D(A_{LC,D} - A_{D,D}) &<0\\
    \Leftrightarrow\quad x_D &>x_{LC}\frac{A_{LC,LC} - A_{D,LC}}{A_{D,D} - A_{LC,D}}\\
    \Leftrightarrow\quad x_D &>x_{LC}\varphi(\gamma,B,m)
\end{align*}
where the final line uses the identity in Equation \eqref{eqn: psi the seperatrix}.
\end{proof}
The evolution of the ratio between the two strategies is only dependent on these strategies, and does not depend on the rest of the strategy space. As such, this plays a critical role in determining when the population will move towards or away from the pure defection equilibrium. We show that the set defined by the intersection of this hyperplane and the simplex is positively invariant, and thus enables further analysis on the existence, uniqueness and basin of attraction of the pure defection strategy.

\begin{proposition}\label{prop: positively invariant}
    Let $v = [0,0,...,-\varphi,1]$. Then the set $U = \{x \in \mathbb{R}^d | \sum^d_k x_k = 1, x_k \geq 0, v \cdot x \leq 0\}$ is positive invariant under the replicator dynamics. 
\end{proposition}
\begin{proof}
    \begin{align*}
        \frac{d}{dt}\sum x_i &= \sum\frac{d}{dt}x_i\\
        &= \sum x_i[(A\mathbf{x})_i -xA\mathbf{x}]\\
        &= 0.\\
        \frac{d}{dt} x_i &= x_i[(A\mathbf{x})_i -xA\mathbf{x}]\\
        &= 0.
    \end{align*}
Let $g(x) = v \cdot x$. The derivative along the boundary: $x_D = \varphi x_{LC}$
    \begin{align*}
        \frac{dg}{dt} &= \dot{x}_D - \varphi \dot{x}_{LC}  \\
        &= x_D[(A\mathbf{x})_D - xA\mathbf{x}] - \varphi x_{LC} [(A\mathbf{x})_{LC} - xA\mathbf{x}]\\
        &= \varphi x_{LC} [(A\mathbf{x})_D - (A\mathbf{x})_{LC}]\\
        &= 0.
    \end{align*}
The last line follows from equality of payoff on the boundary.
\end{proof}
Since $s_D$ is not in the set $U$, the basin of attraction is bounded above by the complement of this set. Specifically, since \eqref{eq: replicator} is Lipschitz, any trajectory such as $v \cdot x \leq 0$ cannot converge to $s_D$. This enables a sufficient condition for cooperation to emerge: the population distribution starts within the set $U$. Importantly, the size of this set is not static. In fact, we can show that the set converges to the entire simplex as the strategy length $m$ increases. As a consequence, the relative basin of attraction for the defection sequence shrinks to zero.
To show this, we show that the probability of sampling anywhere in its basin tends to zero in the limit of the strategy length and discount factor.
\begin{theorem}\label{thm: basin of s_D shrinks to zero}
    For any valid initialisation of $B$, $\lim_{m\rightarrow\infty, \gamma \rightarrow 1} |\text{BoA}(s_D)| = 0$. That is, as strategy length increases the relative size of the basin of attraction shrinks to zero.
\end{theorem}
\begin{proof}
Let $Y_1,...,Y_d \stackrel{i.i.d.}{\sim} \text{Exp}(1)$, with $T = \sum_{i=1}^d Y_i$ and $X_i = \frac{Y_i}{T}$. Then $x =[X_1,...,X_d]$ is uniformly distributed on $\Delta^{d-1}$. The probability of sampling a point inside of $U$ (defined in Proposition \ref{prop: positively invariant}) is:
\begin{align*}
    1-P[X_{s_D} \geq \varphi X_{s_{LC}}]&=1-P[Y_{S_D} \geq \varphi Y_{s_{LC}}]\\
    &= 1-\int_0^\infty P[Y_{s_D}>\varphi t]f_{Y_{LC}}(t) dt\\
    &= 1- \int_0^\infty e^{-\varphi t}e^{-t} dt\\
    &=1-\frac{1}{1+\varphi}
\end{align*}
Since $\lim_{m\rightarrow \infty, \gamma\rightarrow 1}\varphi =\infty $, the probability of sampling inside the set $U$ tends to one. By Proposition \ref{prop: positively invariant}, any trajectory which starts in $U$ will stay in $U$. Hence the relative size of the basin for $s_D \not\in U$ shrinks to zero. 
\end{proof}
This provides an important asymptotic result for avoiding defection: adding degrees of freedom in the strategy space is sufficient to avoid pure defection altogether. In practice, this level of strategic complexity can be unrealistic and computationally infeasible to learn. When shorter finite-length sequences are considered, we provide a comprehensive analysis of the parameter space to analyse the existence and stability of cooperative equilibria. 

\subsection{Parameter Influence on Sequences}
Whilst the separatrix $\varphi$ provides an upper bound for the basin of attraction for $s_D$, it does not provide an exact separation for all basins. In high dimensions, there are many non-linear boundaries. However we can provide an explicit separation for when the sequence length is set to 3, providing the necessary and sufficient conditions for cooperation.
\begin{lemma}\label{lem: m3_exact}
    Let $m=3$. If $T+P\geq 2R$, then $\varphi$ is the exact manifold which separates the simplex. Specifically, $$\lim_{t\rightarrow \infty} \mathbf{x} = 
    \begin{cases}
     D, & x_{D} > x_{LC}\varphi(\gamma,B,3)\\
     \text{Stays  on } \varphi, &  x_{D} = x_{LC}\varphi(\gamma,B,3)\\
     LC, &  x_{D} < x_{LC}\varphi(\gamma,B,3)
    \end{cases}$$
\end{lemma}
\begin{proof}
Since $T+P\geq 2R$, by Lemma \ref{lem: num_nash} there are exactly two Nash equilibria (D) and (LC).  Define $U_D = \{\mathbf{x}\in \Delta : (A\mathbf{x})_D-(A\mathbf{x})_{LC} >0\}$. For all strategies on the interior of the simplex, the log-ratio is given by
\begin{align*}
    \frac{d}{dt}\ln\frac{x_D}{x_{LC}} = (A\mathbf{x})_D - (A\mathbf{x})_{LC},
\end{align*}
which is strictly positive on $U_D$ and zero on the boundary $\varphi$. Note that $\overline{U}_D$ is positively invariant (similar proof to Proposition \ref{prop: positively invariant}). Consider the Lyapunov function $V_D( \mathbf{x}) = \frac{\sum_{i\neq D}x_i}{x_D}$.
\begin{align}
    \nabla V_D \cdot \dot{\mathbf{x}} &= \sum_{i\neq D}\frac{x_i}{x_D}[(A\mathbf{x})_i - \mathbf{x}A\mathbf{x}] - \frac{1-x_D}{x_D^2} x_D [(A\mathbf{x})_D - \mathbf{x}A\mathbf{x}]\\
    &= \sum_{i\neq D}\frac{x_i}{x_D}(A\mathbf{x})_i - \frac{1-x_D}{x_D}(A\mathbf{x})_D \\
    &\leq \sum_{i\neq D}\frac{x_i}{x_D}(A\mathbf{x})_{LC} - \frac{1-x_D}{x_D}(A\mathbf{x})_D \\
    &= \frac{1-x_D}{x_D}(A\mathbf{x})_{LC} - \frac{1-x_D}{x_D}(A\mathbf{x})_D \\
    &= \frac{1-x_D}{x_D}[(A\mathbf{x})_{LC} - (A\mathbf{x})_D] <0.
\end{align}
The third line follows from the dominance of $s_{LC}$ over the other non-defecting strategies, which follows from exploiting any sequence which cooperates before the final step. Equality to zero only occurs on $E=\varphi \cup \{s_D\}$. By LaSalle's invariance principle on $\overline{U}_D$, any trajectory starting within $U_D$ converges to the largest invariant set in $E$. For any strategy $i \neq D$ with $\mathbf{x} \in U_D$,
\begin{align*}
    \frac{d}{dt}\ln\frac{x_D}{x_i} &= (A\mathbf{x})_D - (A\mathbf{x})_{i},\\
    &\geq (A\mathbf{x})_D - (A\mathbf{x})_{LC}\\
    &>0
\end{align*}
which means $x_D(t) \rightarrow 1$ (i.e., the strategy $s_D$). Hence all trajectories starting in $U_D$ converges to $s_D$. The derivation is identical for $V_{LC}( \mathbf{x}) = \frac{\sum_{i\neq LC}x_i}{x_{LC}}$. On the boundary, the derivative of the ratio is exactly zero. Since $\varphi$ is invariant, any trajectory on $\varphi$ remains on $\varphi$.
\end{proof}
This result provides an interesting analytical separation for what a given population will converge to, and prompts questions on the effect of different learning rules on the size of the basin of attraction.
\begin{figure}[H]
    \centering
    \begin{subfigure}[t]{0.4\textwidth}
        \centering
        \includegraphics[width=\textwidth]{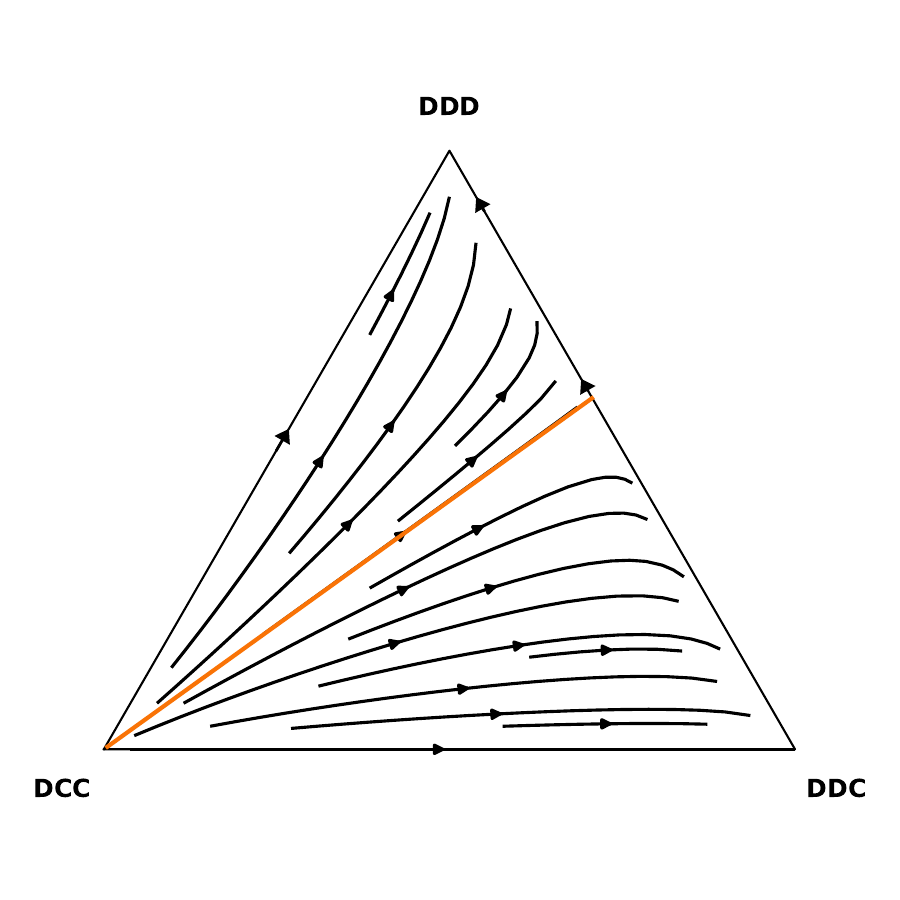}
        \caption{R=3}
        \label{fig: phase_r3}
    \end{subfigure}
    \begin{subfigure}[t]{0.4\textwidth}
        \centering
        \includegraphics[width=\textwidth]{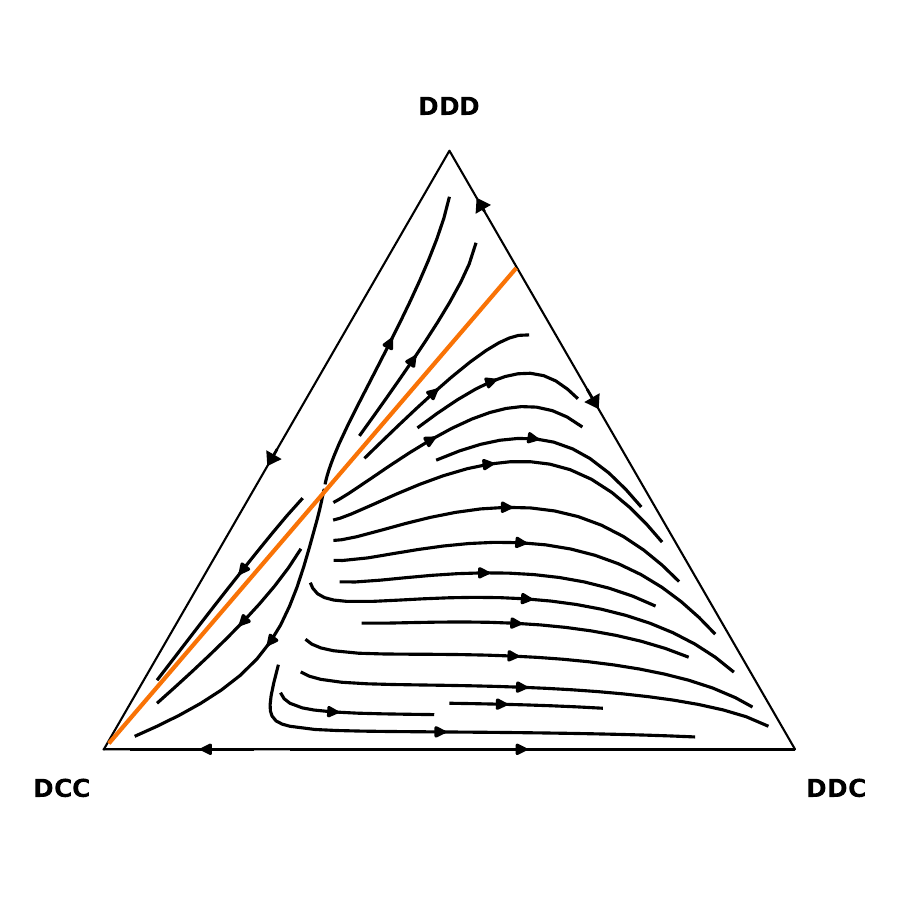}
        \caption{R=4}
        \label{fig: phase_r4}
    \end{subfigure}
    \caption{Phase plot of the projected dynamics for $m=3$. The cases of $T+P\geq 2R$ and $T+P<2R$ are shown; in each, $T=5,P=1$ and $\gamma =0.9$. On the left, a clear linear separation is visible (the manifold $\varphi$ shown in orange), with the simplex divided into two basins. On the right, the addition of another stable sequence adds non-linearity to the vector field.}
    \label{fig: phase}
\end{figure}
The basin separation shown in Figure \ref{fig: phase} switches from a linear to non-linear boundary as the payoff matrix parameter space goes through $T+P=2R$. At this point the number of stable points goes through a critical transition: the number of stable sequences goes from two to three. The dynamics reveal how the separatrix $\varphi$ acts as an upper bound on the basin of attraction for $s_D$. Note that the existence of these cooperative Nash equilibria is also impacted by $\gamma$; to be stable the reward for deviation must be strictly negative, and thus agents must value the final reward of cooperation sufficiently high to avoid early defections.

\begin{figure}[H]
     \centering
\includegraphics[width=.75\textwidth]{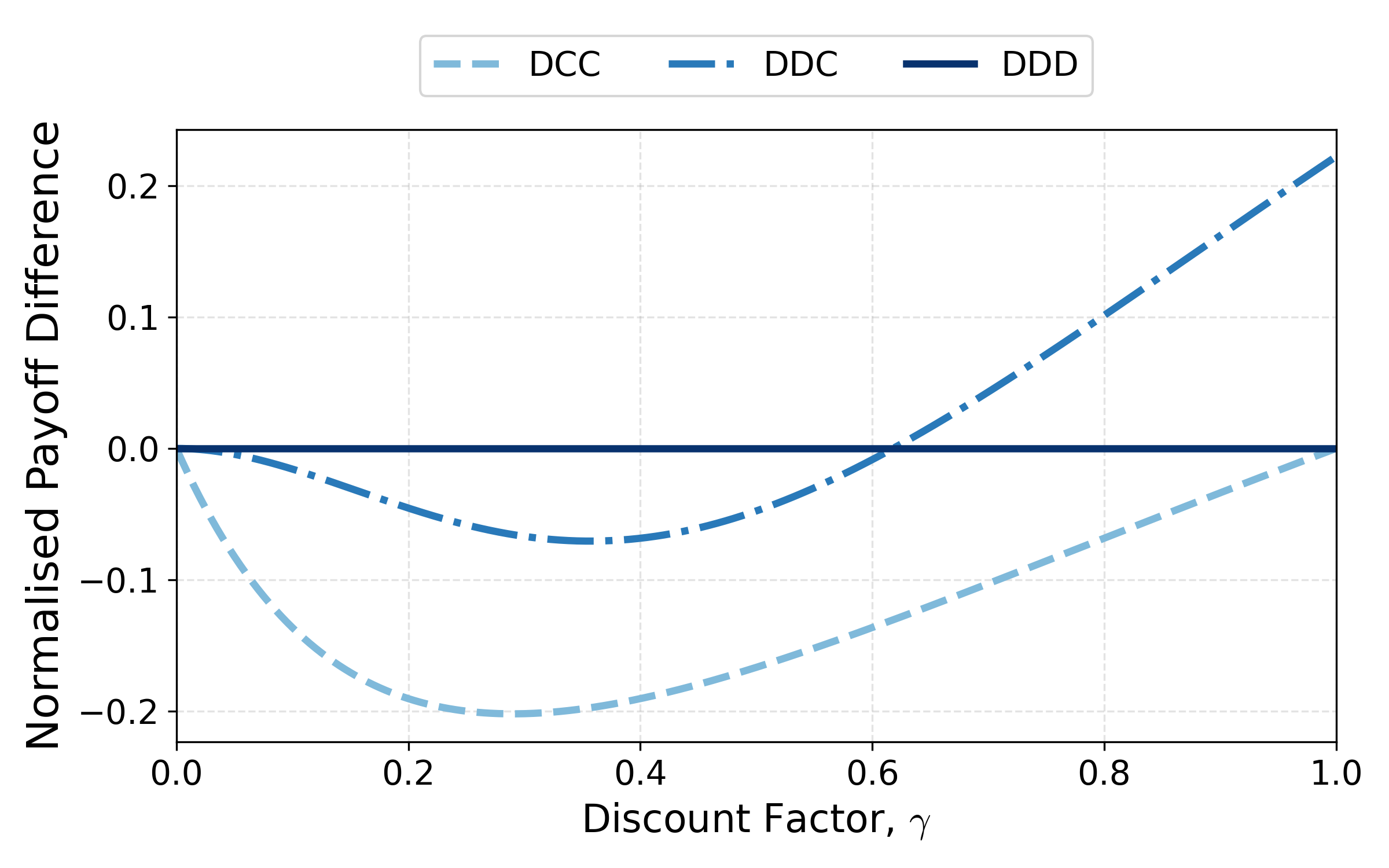}
     \caption{Payoff difference against a DDD population strategy as a function of the discount rate, normalised by the reward for continuing on the sequence. The sequence length is $m=3$, and the base payoff matrix is $R = 3, T = 5, P = 1$. The payoff difference must be non-negative for stability, with a transition through zero indicating a switch as $\gamma$ is varied.}
     \label{fig: reward_vs_gamma}
\end{figure}

As the discount rate is increased, the number of stable sequences increases from one to two, with the stability of the cooperative strategy (DDC) undergoing a bifurcation, shown in Figure \ref{fig: reward_vs_gamma}. This transition occurs at $$\gamma^* = \frac{\sqrt{1+4\frac{T-R}{R-P}}-1}{2},$$ which is the exact root of $\varphi$. Consequently, in Lemma \ref{lem: m3_exact}, $\varphi$ is negative for all $\gamma < \gamma^*$, ensuring the population will converge to the defection sequence for any initial distribution in the interior of the simplex. In higher dimensions, $\varphi$ becomes an $(m-1)$-dimensional polynomial and hence these critical points become non-trivial to find analytically. We can, however, provide an existence guarantee on this transition to stability.
\begin{lemma}\label{lem: psi positive}
        For any initialisation of $B$, there exists an $(m,\gamma) \in \mathbb{N}\times (0,1)$ such that $\forall \gamma \geq \gamma^*,\;\;\varphi(\gamma,B,m) \geq 0$. 
\end{lemma}
\begin{proof}
We first note that since $R-P>0$, $\varphi(\gamma,B,m)$ is monotone increasing in both $m$ and $\gamma$, and continuous in $\gamma$. Consider the left and right limits of the separatrix:
\begin{align*}
    \lim_{\gamma \rightarrow0^+} \varphi &= \frac{R-T}{P} <0,\\
    \lim_{\gamma \rightarrow1^-} \varphi&=\frac{1}{P}\bigg(\lim_{\gamma \rightarrow1^-}\frac{\gamma(1-\gamma^{m-1})}{1-\gamma}(R-P) + (R-T)\bigg).\\
    \intertext{By L'Hopitals Rule,}
    &= \lim_{\gamma \rightarrow1^-}\frac{1-m\gamma^{m-1}}{-1}\frac{(R-P)}{P} + \frac{(R-T)}{P},\\
    &= \frac{(m-1)(R-P ) + (R-T)}{P}.
\end{align*}
For a root to exist in $(0,1)$, $\varphi(1^-) >0$. Therefore, let \begin{equation*}
    m> \frac{T-R}{R-P} +1,
\end{equation*}
then by intermediate value theorem there exists a $\gamma \in (0,1)$ with $\varphi(\gamma,B,m) = 0$. The result follows from the monotonicity of $\varphi$ in $\gamma$.
\end{proof}

Since $\varphi$ creates a separation of basins on the $s_D$ and $s_{LC}$ face of the simplex, it provides an explicit condition for the negativity of the largest eigenvalue at the pure strategy $s_{LC}$, dictating the stability of this minimal cooperating sequence. In the infinite limit of the sequence length, the critical point becomes
\begin{equation*}
    \gamma^* = \frac{T-R}{T-P}.
\end{equation*} One can easily verify that no cooperative sequence can be stable unless $s_{LC}$ is also stable, hence this provides a global necessary and sufficient condition on $\gamma$ for the existence of stable sequences beyond that of pure defection. 
\begin{lemma}\label{lemma: stable iff LC stable}
    A stable cooperative sequence exists if and only if the last step cooperator $s_{LC}$ is stable.
\end{lemma}
\begin{proof}
    If $s_{LC}$ is stable then trivially a stable cooperative sequence exists. For the backwards direction, suppose that there is a stable cooperative sequence denoted by $s\neq s_D$, whilst $s_D$ is unstable. In particular, if $s_D$ is unstable then $\varphi(\gamma,B,m) <0$ which is equivalent to
    \begin{align*}
        (R-P)\sum_{j=1}^{m-1} \gamma^{j} &< T-R.
    \end{align*}
    Now consider the strategy $s$: there exists a $k \in [0,m-1]$ such that $s_k = C$. Since $s$ is stable, the payoff $A(s,s)>A(s,s')$ for all $s' \in S$. In particular, a payoff upper bound of $A(s,s)$ must beat a mutant strategy which defects on the $k$th step:
    \begin{align*}
        \sum_{i=1}^{k-1}\gamma^{i-1}u_i + \frac{\gamma^{k-1}R}{1-\gamma} &> \frac{\sum_{i=1}^{k-1}\gamma^{i-1}u_i + \gamma ^{k-1}T}{1-\gamma^{k}}\\
       \Leftrightarrow\quad\;\qquad \frac{1-\gamma^{k-1}}{1-\gamma}R &> \gamma\frac{1-\gamma^{k-1}}{1-\gamma}P + T\\
       \Leftrightarrow\qquad(R-P) \sum_{j=1}^{k-1}\gamma^j &> T-R\\
       \Rightarrow\qquad(R-P) \sum_{j=1}^{m-1}\gamma^j &> T-R
    \end{align*}
    which contradicts $s_D$ being unstable.
\end{proof}
So far, we have focused on specific sequences particularly around avoiding defection: Lemma \ref{lem: m3_exact} raises the important point of mixed Nash equilibria and Lemma \ref{lem: psi positive} presents conditions for stability for the LC strategy. This provides a necessary and sufficient condition for the stability of cooperative sequences, shown in Lemma \ref{lemma: stable iff LC stable}. We now turn our attention to more general stability properties of the system.

\section{The Structure and Stability of Stable Sequences}
\label{sect: theory2}
Generally, repeated games can display interesting and stable mixed Nash equilibria. In the Prisoner's Dilemma, no such mixed strategy exists. It is therefore not surprising that in this repeated game with trigger restarts, any mixed Nash equilibrium (NE) with the pure defection strategy in its support is also evolutionarily unstable.

\begin{lemma}\label{lem: mixed nash are not stable}
    Let $p$ be a mixed NE such that $p_{D} >0$. Then $p$ is not an evolutionarily stable strategy (ESS).
\end{lemma}
\begin{proof}
    Suppose for the sake of contradiction that $p$ is ESS. Since $s_D$ is in the support, and every pure strategy in the support attains the same payoff against $p$,
    \begin{align*}
        u(s_D,p) = u(p,p).
    \end{align*}
    Take strategy $q = s_D$, then $u(q,p) = u(p,p)$, and we require the second ESS condition
    \begin{align*}
        u(p,q) &> u(q,q)\\
        \Rightarrow u(p,s_D) &> u(s_D,s_D)
    \end{align*}
    but the strategy $s_D$ cannot be exploited, since playing any move other than defection is met with the Sucker's payoff and the sequence restarts. Hence, $u(p,s_D) \leq  \frac{P}{1-\gamma} =u(s_D,s_D)$. This is a contradiction, so $p$ is not ESS.
\end{proof}
As a result, every Nash equilibrium in the interior of the simplex is not an ESS. Since each sequence defines a vertex of the simplex $\Delta$, it is important to check how this impacts the stability. We consider exponential stability, where any perturbation is diminished exponentially fast. In the context of replicator dynamics, we can show that all strict Nash equilibria are exponentially stable.
\begin{proposition}\label{prop: strict nash are exponentially stable}
    Let $p$ be a strict Nash equilibrium of (\ref{eq: replicator}). Then p is exponentially stable.
\end{proposition}
\begin{proof}
      Suppose $p$ is the pure strategy $x_i =1$. Then the eigenvalues $\lambda_j = A_{ji}-A_{ii}$ for all $j \notin \text{supp}(p)$. Since $p$ is a strict Nash, $\lambda_j < 0\quad \forall j \in d \setminus i$. The feasible perturbations lie in the tangent space of the simplex. Therefore, for the dynamics on the tangent space, the Jacobian is Hurwitz and thus exponentially stable.
\end{proof}
This holds true for all systems, not just the repeated game we consider. In this setting, all stable sequences are strict Nash equilibria. Therefore, as a direct corollary, all the stable sequences are exponentially stable. Thus far, we have focused on the specific defection (D) and last-step cooperator strategies (LC) and shown global exponential stability for all stable sequences. We turn our attention to how the whole system evolves, and what structural properties appear in the sequences and distribution of their sizes.

\subsection{Sequence Structure and Distribution}
To analyse the basin of attraction distribution, we first provide an explicit number of stable sequences for any given parametrisation of the payoff matrix. This enables analysis on the structure of stable sequences.
\begin{lemma}\label{lem: num_nash}
    Let $ m\geq 2, \gamma \rightarrow 1$, and $\kappa = \lfloor \frac{T-R}{R-P}\rfloor+2$. For $m\geq \kappa$, there are exactly $1+2^{m-\kappa}$ pure Nash equilibria. If $m < \kappa$, defection is the only Nash Equilibrium.
\end{lemma}
\begin{proof}
    Let $\delta = T-R$ be the largest marginal value of deviation in a single step and $u_i$ the payoff in step $i$ of a sequence. For a sequence to be stable then $\sum_{i=1}^k u_i + \delta \leq ku_m, \;\; \forall k \in[1,m]$. There are two cases to consider: $u_m =  P$ or $u_m = R$.
    \begin{itemize}
        \item \textbf{Case 1}: $u_m = P$. Suppose not all rewards are $P$. Then to be stable $$kP +\delta \leq\sum_{i=1}^k u_i + \delta < kP$$ which does not holds for any $\delta >0$. Therefore the only stable sequence is if $u_k = P\;\; \forall k \in[1,m]$.
        \item \textbf{Case 2}: $u_m = R$. Let $n_d(k)$ and $n_c(k)$ be the number of Ds and Cs played up to and including step $k$. Then if $n_c(k) =0$, there is no positive deviation and the sequence is stable up to step $k$. Otherwise, $n_c(k) >0$ and the requirement is       
        \begin{align*}
            n_c(k)R + n_d(k)P + \delta &< kR\\
            \delta &< n_d(k)(R-P)\\
            \frac{\delta}{R-P} &< n_d(k)
        \end{align*}
        Since this must hold for all $k$, the first $\lfloor \frac{\delta}{R-P}\rfloor+1$ must be action D. All other actions are arbitrary (except the final one); so there are $2^{m-\kappa}$ sequences ending in $R$, provided $m \geq \kappa$. 
    \end{itemize} 
\end{proof}
The restriction on $m\geq2$ is imposed to obtain an exact formula for the number of stable sequences, with defection as the only Nash equilibrium for $m=1$. The proof steps provide explicit requirements on the matrix for cooperative sequences to be stable. Indeed, the second case highlights a key property of stable sequences: the hazing period. This defines the number of defections which occur at the beginning of the sequence for cooperation to be stable. Intuitively, this must be long enough to make future deviations in the sequence sub-optimal, so as to avoid repeating the hazing. In the 2-action Prisoner's Dilemma, this hazing period becomes a necessary condition and, if the final action is to cooperate, it is also a sufficient condition. This counting provides two points of knowledge: the exact structure of all stable sequences, and the number of basins of attraction. The latter could then be used to inform the error when estimating the basin of attraction sizes.
\begin{figure}[H]
     \centering
\includegraphics[width=.85\textwidth]{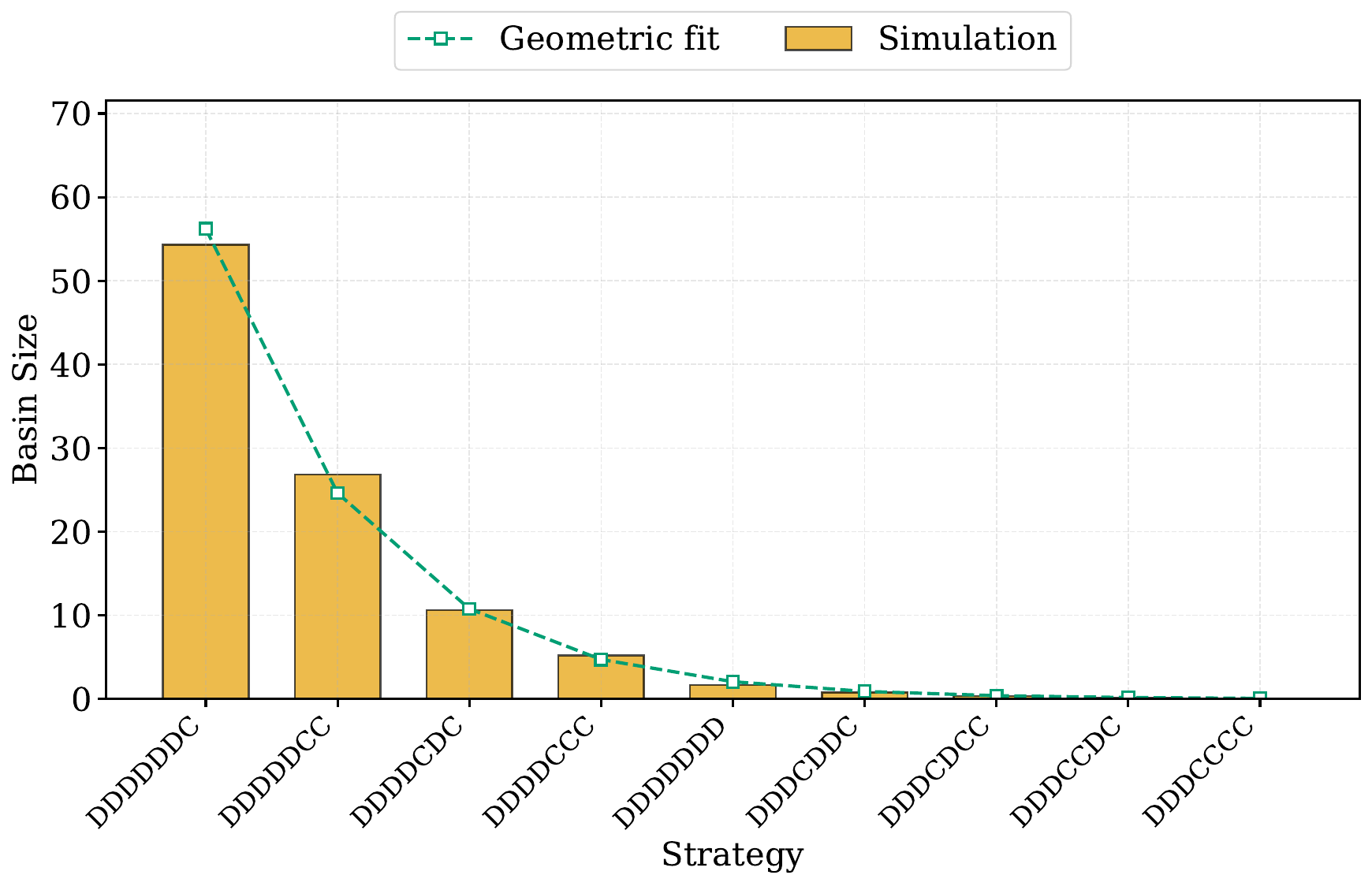}
     \caption{Simulated distribution of BoA sizes for $m=7$. A geometric fit is plotted to illustrate a possible distribution of the basins. 100,000 initial conditions are used to estimate the basin sizes.}
     \label{fig: dist_nash_boA}
\end{figure}

As shown in Figure \ref{fig: dist_nash_boA}, the basin sizes resemble a geometric distribution which is consistent with potential games \cite{collevecchio2024basins}. As shown earlier, however, the PD with restarts is not a potential game (Lemma \ref{lem: potential game}). Without a global potential function, we cannot provide an explicit structure without significant assumptions and simplifications. We can, however, observe that the largest basins have longer hazing periods. This is largely a result of the high dimensionality: there are many stable sequences containing a mix of defection and cooperation actions. It is hard for a more cooperative sequence to compete with a population where defections are common. This is particularly true for the optimal sequence, which maximises the total reward by cooperating after the minimal hazing length. 

As a result, convergence to optimality is only possible for very short sequence lengths, when defection possibilities are limited. Figure \ref{fig: optimal} illustrates the estimated basin size as the sequence length $m$ increases. Intuitively, the drop can be understood as follows: there is only one social optimal sequence; there are $2^{m-\kappa}$ other stable sequences each of which will exploit this optimality by defecting at some time-step. 

\begin{figure}[H]
     \centering
    \includegraphics[width=.75\textwidth]{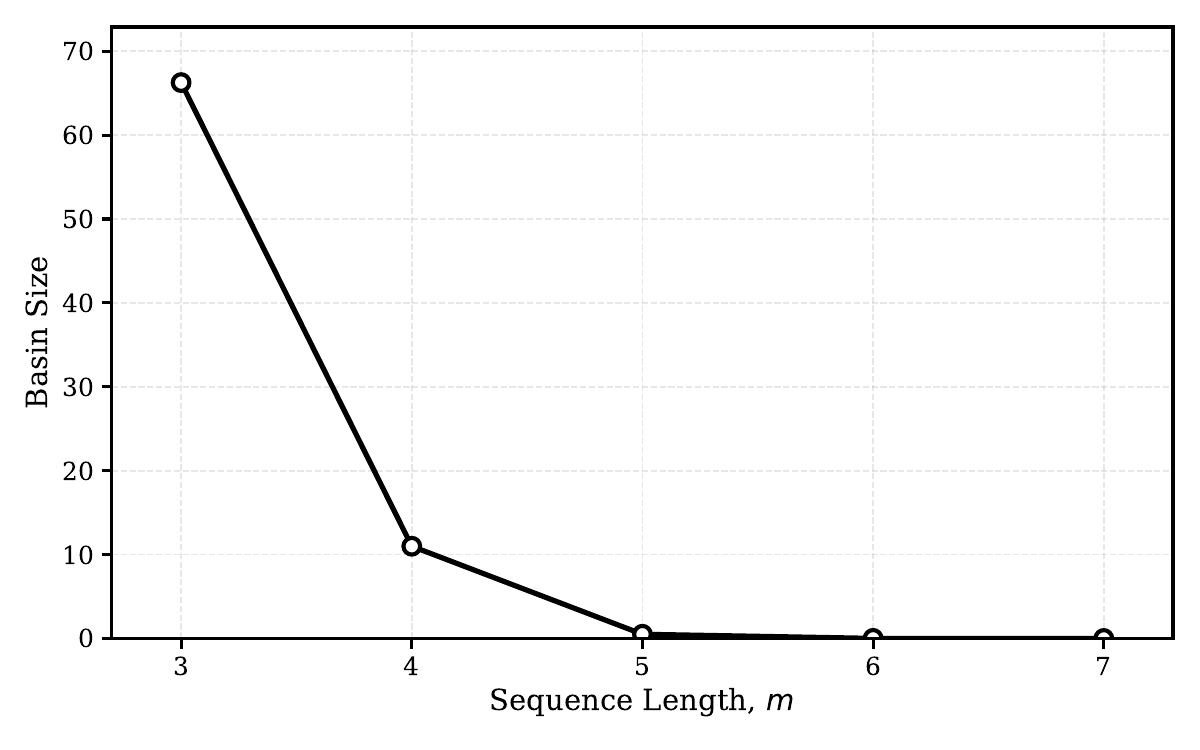}
     \caption{Basin of Attraction of the Optimal Sequence, which is the stable sequence which maximises the payoff. For each, $R=3, T=5, P=1$ and $\gamma=0.99$. The basin is calculated using $100,000$ simulated points.}
     \label{fig: optimal}
\end{figure}

\subsection{Imperfect Restarts}

Thus far, we have considered a perfect restart mechanism, in which agents' sequences are reset upon action misalignment. In this section, we analyse how stable sequences emerge through evolutionary pressure when there are imperfect restarts. In particular, we consider when the interaction restarts with probability $\varepsilon$ even when actions agree. This challenges the goal of stable cooperation, where the hazing period can detract from previously stable cooperative pairings.

The error rate of $\varepsilon$ fundamentally alters the payoff structure in the normal form game. At each step, when actions agree, the sequence continues with probability $(1-\varepsilon)$. Define $q=\gamma(1-\varepsilon)$, then the payoff is (see Appendix \ref{secA1})
\begin{equation}\label{eq: payoff_nf imperfect}
A^\varepsilon(s, t) :=
\begin{cases}
\dfrac{\displaystyle (1-q)\sum_{i=1}^{\tau} q^{\,i-1}\,B_{s_i,t_i}}
{(1-\gamma)(1 - q^{\tau})}, 
& \tau(s,t)<\infty, \\[1em]
{\displaystyle \frac{(1-q)\sum_{i=1}^{m-1} q^{i-1} B_{s_i,s_i} + q^{m-1} B_{s_m,s_m}}{1-\gamma}},
& \tau(s,t)=\infty.
\end{cases}
\end{equation}
We note that $\varepsilon=0$ recovers the original payoff in \eqref{eq: payoff_nf}. This structure allows us to re-derive the separatrix, and directly compare against the perfect restart mechanism. 

\begin{align*}
    \varphi^\varepsilon(\gamma,B,m)
    &= \frac{(R-P)\sum_{j=1}^{m-1}q^j+(R-T)}{P}.
\end{align*}
This is the exact separatrix derived for the deterministic system such that $\varphi^\varepsilon(\gamma,B,m) = \varphi(\gamma(1-\varepsilon),B,m)$. As a consequence, the separation defined in Lemma \ref{lem: separatrix splits Ax} and invariance induced by the separatrix in Proposition \ref{prop: positively invariant} follow immediately. However, the large $m$ behaviour can change significantly. As $m \rightarrow \infty, \gamma \rightarrow 1$, the upper bound on the basin of attraction for the strategy $s_D$ does not collapse to zero. The large $m$ limit of the separatrix is
\begin{align*}
    \lim_{m\rightarrow \infty, \gamma \rightarrow1}\varphi^\varepsilon &= \frac{(R-P)\frac{1-\varepsilon}{\varepsilon}+(R-T)}{P},
\end{align*}
which is finite for all $\varepsilon >0$. Therefore, since Proposition \ref{prop: positively invariant} holds for $\varphi^\varepsilon$, the same proof in Theorem \ref{thm: basin of s_D shrinks to zero} yields a volume bounded above from zero. As such, we cannot ensure that pure defection will leave the population. 

We can however, find the critical error rate for which cooperative equilibria exist. Recall the positivity of the separatrix is a necessary requirement for any stable sequence to emerge. Therefore, we can equate the separatrix to zero to find a critical $\varepsilon$.

\begin{proposition}\label{prop: restart probability bound for existence of cooperative sequence}
    Let $\varepsilon>0$ be the imperfect restart probability. For sufficiently large $m$ and $\gamma$, there is at least one cooperative sequence if and only if
    \begin{align*}
        \varepsilon < \frac{R-P}{T-P}.
    \end{align*}
\end{proposition}
\begin{proof}
    By Lemma \ref{lemma: stable iff LC stable}, there exists a stable cooperative sequence if and only if $s_{LC}$ is stable. Hence it suffices to categorise the stability of the $s_{LC}$ strategy. By Lemma \ref{lem: psi positive}, there exists a $\gamma^*$ such that for all $\gamma \geq \gamma^*, \varphi(\gamma,B,m)\geq0$. In particular, consider $m \rightarrow \infty$, then without imperfect restarts the critical value is
    \begin{align*}
        \gamma^* = \frac{T-R}{T-P}
    \end{align*}
    In the case of imperfect restarts, we use the identity $\varphi^\varepsilon(\gamma,B,m) =\varphi(\gamma(1-\varepsilon),B,m)$. The condition is then
    \begin{align*}
        \gamma^* = \frac{1}{1-\varepsilon}\frac{T-R}{T-P}.
    \end{align*}
    For sufficiently large $\varepsilon$, $\gamma^*$ will exceed 1 and therefore be invalid. This happens when 
    \begin{align*}
        (1-\varepsilon) &= \frac{T-R}{T-P}\\
        \Leftrightarrow\quad\varepsilon &= \frac{R-P}{T-P}.
    \end{align*}
    Hence for all $\varepsilon < \frac{R-P}{T-P}$, there exists an $m^* \in \mathbb{N}$ and a $\gamma^* \in (0,1)$ such that for $m>m^*$ and $\gamma>\gamma^*$ the strategy $s_{LC}$ is stable. For $\varepsilon \geq \frac{R-P}{T-P}$, then there is no $\gamma \in (0,1)$ such that $\lim_{m\rightarrow\infty}\varphi(\gamma(1-\varepsilon),B,m)\geq0$. Since $\varphi^\varepsilon$ is increasing in $m$, $\varphi^\varepsilon<0$ for any $m$ and therefore $s_{LC}$ is unstable. 
\end{proof}

The critical region for the stability of a cooperative sequence is shown in Figure \ref{fig: critical region for stability}. Both a sufficiently high discount factor and a sufficiently low restart probability are required for any such sequence to be stable. By increasing the sequence length $m$, the longer hazing period in $s_{LC}$ expands this critical region. This creates a challenge between the large-$m$ sequences gaining stability in error-prone relations, and the optimal cooperative sequence having a larger basin size for small-$m$ (Figure \ref{fig: optimal}). This captures how short and simple sequences can be the easiest for learning populations to adopt, whilst longer sequences require more coordination but offer additional stability. 
\begin{figure}[h]
     \centering
\includegraphics[width=.75\textwidth]{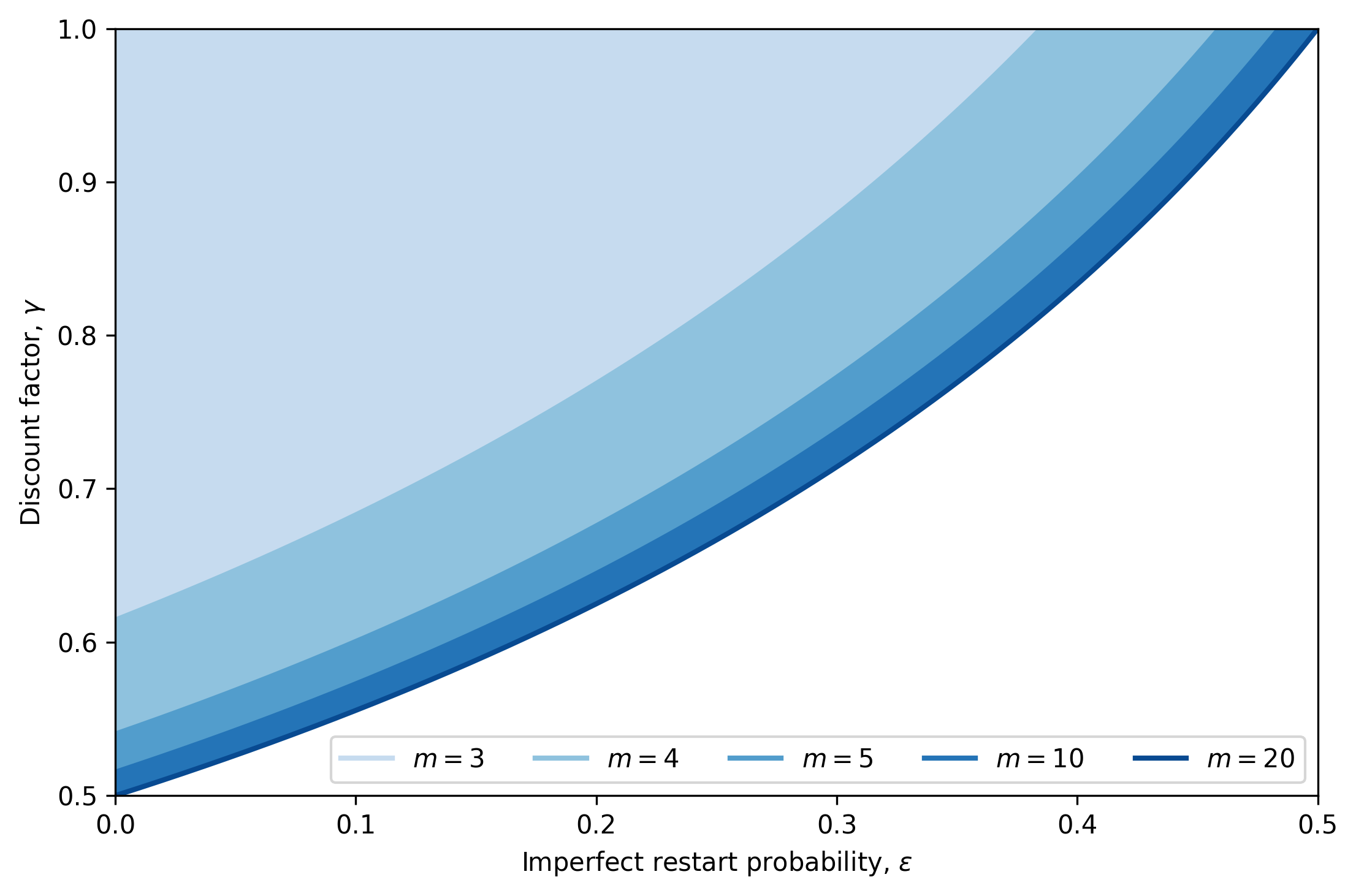}
     \caption{Critical region for $s_{LC}$ to be stable as $\gamma$ and $\varepsilon$ are varied. Increasing the strategy length increases the feasible region for stability. The payoff is such that $R=3,T=5,P=1$.}
     \label{fig: critical region for stability}
\end{figure}

\section{Discussion}
\label{sect: discussion}
We have shown that increasing strategic complexity can yield cooperative, stable equilibria when coupled with the trigger-restart mechanism. Unlike cooperation-favouring partner selection rules in the literature such as Out-for-Tat \cite{Zhang2016, zheng_simple_2017}, this mechanism merely resets a relationship when the agents disagree. Yet, we show that pure defection can still be eliminated from the population. We advance the literature showing how a population of agents can learn the theoretically stable sequences, and provide necessary and sufficient conditions on the parameter space for stability to be guaranteed. We find an analytical form for the `hazing period', where agents must defect before cooperation can emerge.

Moreover, we demonstrate that this stability is robust to execution errors, where agents face noisy observations or unintended deviations. Our analysis of imperfect restarts highlights the trade-off between the ease of learning simple cooperative sequences and the error-tolerance provided by longer, more complex sequences.

While our approach provides strong analytical guarantees, we note some limitations. Our results rely on replicator dynamics, which assume an infinite, well-mixed population. In finite populations with stochastic interactions, drift and exploration noise could alter the sizes of the basins of attraction and the time required to converge to cooperative sequences.

Future research could extend this paradigm by relaxing the strict trigger-restart mechanism, allowing agents to co-learn the restart conditions alongside their action sequences. For instance, enabling agents to dynamically adjust their resetting threshold could allow them to dynamically adapt to the varying strategic profile in the environment.
We hope that future work can build upon our analysis to further explain how mechanics affecting local interactions can induce cooperation in complex, learning populations.

\bibliography{sn-bibliography}
\begin{appendices}

\section{Additional Derivations}\label{secA1}

We present the full algebraic derivation of the manifold between the two strategies $s_D$ and $s_{LC}$.
\begin{align}
    \varphi(\gamma,B,m) &= \frac{A_{LC,LC}-A_{D,LC}}{A_{D,D}- A_{LC,D}}\\
     &=\frac{\sum_{j=0}^{m-2}(P\gamma^{j})+\frac{R\gamma^{m-1}}{1-\gamma} - \frac{\sum_{j=0}^{m-2}(P\gamma^{j})+T\gamma^{m-1}}{1-\gamma^{m}}}{P/(1-\gamma)-\frac{\sum_{j=0}^{m-2}(P\gamma^{j})}{1-\gamma^{m}}}\\
     &=  \frac{P(1-\gamma^{m-1})(1-\gamma^{m}) + R\gamma^{m-1}(1-\gamma^{m}) - P(1-\gamma^{m-1}) -T\gamma^{m-1}(1-\gamma)}{P(1-\gamma^{m}) - P(1-\gamma^{m-1})}\\
     &=  \frac{P(1-\gamma^{m-1})(-\gamma^{m}) + R\gamma^{m-1}(1-\gamma^{m}) -T\gamma^{m-1}(1-\gamma)}{P\gamma^{m-1}(1-\gamma)}\\
     &=  \frac{P(1-\gamma^{m-1})(-\gamma) + R(1-\gamma^{m}) -T(1-\gamma)}{P(1-\gamma)}\\
     &=  \frac{P(1-\gamma)-P(1-\gamma^{m}) + R(1-\gamma^{m}) -T(1-\gamma)}{P(1-\gamma)}\\
     &= \frac{P-T}{P} + \frac{R-P}{P}\frac{1-\gamma^{m}}{1-\gamma}\\
    &= \frac{(R-P)\sum^{m-1}_{j=0}\gamma^j - (T-P)}{P}\\
    &= \frac{(R-P)\sum^{m-1}_{j=1}\gamma^j + (R-T)}{P}.
\end{align}

\subsection{Imperfect Restarts}

Suppose the first deviation in the sequence happens at $\tau < \infty$, then the payoff for the row player is
\begin{align*}
    A^\varepsilon(s,t) &= \sum_{i=1}^\tau \gamma^{i-1}(1-\varepsilon)^{i-1} B_{s_i,t_i} + \sum_{i=1}^{\tau-1}\gamma^{i-1}(1-\varepsilon)^{i-1}\varepsilon \gamma A^\varepsilon(s,t) + \gamma^\tau(1-\varepsilon)^{\tau-1}A^\varepsilon(s,t)\\
    \intertext{Let $q = \gamma(1-\varepsilon)$. Rearranging for $A^\varepsilon$ yields}
    A^\varepsilon(s,t) &= \frac{\sum_{i=1}^\tau q^{i-1} B_{s_i,t_i}}{1-\varepsilon\gamma\sum^{\tau-1}_{i=1}q^{i-1} - \gamma q^{\tau-1}}.
\end{align*}
The denominator simplifies to
\begin{align*}
    1-\varepsilon\gamma\sum^{\tau-1}_{i=1}q^{i-1} + \gamma q^{\tau-1} &= 1- \frac{\varepsilon\gamma(1-q^{\tau-1})}{1-q} -\gamma q^{\tau-1}\\
    &= \frac{1-q-\varepsilon\gamma(1-q^{\tau-1})-\gamma q^{\tau-1}(1-q)}{1-q}\\
    &= \frac{(1-\varepsilon\gamma -q) -\gamma q^{\tau-1}(1-\varepsilon) + \gamma q^{\tau}}{1-q}\\
    &= \frac{1-\gamma - q^\tau(1-\gamma)}{1-q}\\
    &= \frac{(1-\gamma)(1-q^\tau)}{1-q},
\end{align*}
which yields the expression
\begin{align*}
    A^\varepsilon(s,t) &= \frac{(1-q)\sum_{i=1}^\tau q^{i-1} B_{s_i,t_i}}{(1-\gamma)(1-q^\tau)}.
\end{align*}
When no deviation occurs ($\tau = \infty$), the sequence only restarts according to the error rate. This corresponding payoff is
\begin{align*}
    A^\varepsilon(s,t) &= \sum_{i=1}^{m-1} q^{i-1} B_{s_i,s_i} + \sum_{i=m}^{\infty}q^{i-1}B_{s_m,s_m} + \sum_{i=1}^\infty q^{i-1}\varepsilon\gamma A^\varepsilon(s,t)\\
    &= \sum_{i=1}^{m-1} q^{i-1} B_{s_i,s_i} + \frac{q^{m-1}}{1-q}B_{s_m,s_m} + \frac{\varepsilon\gamma}{1-q}A^\varepsilon(s,t)\\
    \intertext{Rearranging for $A^\varepsilon$ yields}
    A^\varepsilon(s,t) &= \frac{(1-q)\sum_{i=1}^{m-1} q^{i-1} B_{s_i,s_i} + q^{m-1} B_{s_m,s_m}}{1-\gamma}.
\end{align*}
In summary,
\begin{equation*}
A^\varepsilon(s, t) :=
\begin{cases}
\dfrac{\displaystyle (1-q)\sum_{i=1}^{\tau} q^{\,i-1}\,B_{s_i,t_i}}
{(1-\gamma)(1 - q^{\tau})}, 
& \tau(s,t)<\infty, \\[1em]
{\displaystyle \frac{(1-q)\sum_{i=1}^{m-1} q^{i-1} B_{s_i,s_i} + q^{m-1} B_{s_m,s_m}}{1-\gamma}},
& \tau(s,t)=\infty.
\end{cases}
\end{equation*}
The payoff for the 4 strategies are
\begin{align*}
   A_{LC,LC} &=  \frac{(1-q)\sum_{i=1}^{m-1} Pq^{i-1} + Rq^{m-1}}{1-\gamma} =\frac{P(1-q^{m-1}) + Rq^{m-1}}{1-\gamma}\\
  A_{D,LC} &=   \frac{(1-q)\sum_{i=1}^{m-1} Pq^{\,i-1} + (1-q)q^{m-1}T}{(1-\gamma)(1 - q^{m})} = \frac{P(1-q^{m-1}) + T(1-q)q^{m-1}}{(1-\gamma)(1-q^m)}\\
 A_{D,D} &=  \frac{P}{1-\gamma} \\
   A_{LC,D} &=   \frac{(1-q)\sum_{i=1}^{m-1} Pq^{\,i-1}}{(1-\gamma)(1 - q^{m})} = \frac{P(1-q^{m-1})}{(1-\gamma)(1-q^m)}.
\end{align*}
The deviation away from the pure strategies is
\begin{align*}
     A_{LC,LC} -  A_{D,LC} &= \frac{q^{m-1}}{(1-\gamma)(1-q^m)}\big[(R-T)(1-q) + (R-P)(q-q^m)\big]\\
      A_{D,D} -  A_{D,LC} &= \frac{Pq^{m-1}(1-q)}{(1-\gamma)(1-q^m)}.
\end{align*}
Therefore,
\begin{align*}
    \varphi^\varepsilon(\gamma,B,m)& = \frac{A_{LC,LC} -  A_{D,LC}}{ A_{D,D} -  A_{D,LC}}\\
    &= \frac{(R-T)(1-q) + (R-P)(q-q^m)}{P(1-q)}\\
    &= \frac{(R-P)\sum_{j=1}^{m-1}q^j+(R-T)}{P}.
\end{align*}

\end{appendices}

\end{document}